\documentclass{amsart}
\usepackage{amsfonts}
\usepackage{blkarray}
\usepackage{hyperref}
\usepackage{float,multirow}
\newtheorem{theorem}{Theorem}[section]
\newtheorem{lemma}[theorem]{Lemma}

\theoremstyle{definition}
\newtheorem{definition}[theorem]{Definition}
\newtheorem{example}[theorem]{Example}

\theoremstyle{remark}
\newtheorem{remark}[theorem]{Remark}
\numberwithin{equation}{section}

\begin{document}
\title[New extremal binary self-dual codes from codes over $R_{k,m}$]{New
extremal binary self-dual codes of lengths 66 and 68 from codes over $%
R_{k,m} $}
\author{Ab\.id\.in Kaya}
\author{Nes\.ibe T\"ufek\c{c}\.i}
\address{Department of Mathematics, Fatih University, 34500, Istanbul, Turkey%
}
\email{akaya@fatih.edu.tr,nesibe.tufekci@fatih.edu.tr}

\begin{abstract}
In this work, four circulant and quadratic double circulant (QDC)
constructions are applied to the family of the rings $R_{k,m}$.
Self-dual binary codes are obtained as the Gray images of self-dual
QDC codes over $R_{k,m}$. Extremal binary self-dual codes of length
64 are obtained as Gray images of $\lambda$-four circulant codes
over $R_{2,1}$ and $R_{2,2}$. Extremal binary self-dual codes of
lengths 66 and 68 are constructed by applying extension theorems to
the $\mathbb{F}_{2}$ and $R_{2,1}$ images of these codes. More
precisely, 11 new codes of length 66 and 39 new codes of length 68
are discovered. The codes with these weight enumerators are
constructed for the first time in literature. The results are
tabulated.
\end{abstract}

\maketitle

\section{Introduction}

An interesting family of linear codes are self-dual codes.\
Self-dual codes over finite fields have been studied extensively.
Some good binary codes such as the extended binary Golay code and
the extended quadratic residue
codes of parameters $\left[ 48,24,12\right] _{2}$ and $\left[ 104,52,20%
\right] _{2}$ are of this type. Such codes have also attracted a lot
of attention due to their connections to design theory.

Conway and Sloane gave an upper bound for the minimum distance of a
binary self-dual code in \cite{conway}. The bound was finalized in
\cite{rains} as follows; the minimum distance $d$ of a binary
self-dual code of length $n$ satisfies $d\leq 4\left[ n/24\right]
+6$ if $n\equiv 22\pmod{24}$ and $d\leq 4\left[ n/24\right] +4$,
otherwise. A self-dual code meeting this bound is called
$\emph{extremal}$. The possible weight enumerators of extremal
self-dual binary codes of lengths up to $64$ and $72$ were
determined in \cite{conway}. Since then, constructing new extremal
binary self-dual codes have been an attractive research area.
Different techniques such as circulant constructions, automorphism
groups and extensions are used to obtain new extremal binary
self-dual codes. For some of the works done in this direction we
refer the reader to \cite{dougherty,kim,tsai,yankov}.

Recently, some rings of characteristic $2$ have been used
effectively to construct new extremal binary self-dual codes. Lifts
were used in \cite{karadeniz68} and \cite{kaya}. Extension theorems
for self-dual codes were applied to codes over
$\mathbb{F}_{4}+u\mathbb{F}_{4}$ in \cite{kayayildiz}. Karadeniz et
al. used four circulant construction over
$\mathbb{F}_{2}+u\mathbb{F}_{2}$ in \cite{karadeniz}.

In this work, we give a generalization of four circulant
construction and combine the lifting and extending methods. The
computational algebra system MAGMA \cite{magma} is used for the
results. The rest of the paper is organized as follows: Section 2
consists of preliminaries about the family of rings $R_{k,m}$ and
codes over these. In Section 3, we introduce quadratic double
circulant codes over $R_{k,m}$. Section 4 includes constructions for
extremal singly-even binary self-dual
codes of length $64$ as Gray images of four circulant self-dual codes over $%
R_{2,1}$ and $R_{2,2}$. In Section 5, extremal binary self-dual
codes of lengths $66$ and $68$ with previoulsy unknown weight
enumerators are
constructed as extensions and as Gray image of extensions. More precisely, $%
11$ new codes of length $66$ and $39$ new codes of length $68$ are
constructed.

\section{Preliminaries\label{preliminaries}}

The ring $R_{k,m}$ was introduced in \cite{tufekci} as a generalization of $%
\mathbb{F}_{2}+u\mathbb{F}_{2}+v\mathbb{F}_{2}+uv\mathbb{F}_{2}$,
which was studied in \cite{yildizr2}. The ring is a commutative
local Frobenius ring of characteristic 2 that is defined as
\begin{equation*}
R_{k,m}=\mathbb{F}_{2}[u,v]/\left\langle
u^{k},v^{m},uv-vu\right\rangle \text{ where }k\geq m\geq 1.
\end{equation*}

Note that $R_{2,2}=\mathbb{F}_{2}+u\mathbb{F}_{2}+v\mathbb{F}_{2}+uv\mathbb{F%
}_{2}$, for more details on the structure of the ring we refer to \cite%
{tufekci}.

A linear $\mathcal{C}$ of length $n$ over $R_{k,m}$ is an-$R_{k,m}$
submodule of $R_{k,m}^{n}$. The dual $\mathcal{C}^{\perp }$ of a
linear code
$\mathcal{C}$ is defined with respect to the Euclidean inner product as%
\begin{equation*}
\mathcal{C}^{\perp }:=\left\{ (b_{1},b_{2},\ldots b_{n})\in
R_{k,m}^{n}\mid \sum\limits_{i=1}^{n}a_{i}b_{i}=0,\forall
(a_{1},a_{2},\ldots a_{n})\in \mathcal{C}\right\} .
\end{equation*}%
A code $\mathcal{C}$ is said to be \emph{self-orthogonal} if $\mathcal{C}%
\subseteq \mathcal{C}^{\perp },$ and \emph{self-dual} if $\mathcal{C}=%
\mathcal{C}^{\perp }$. A binary self-dual code is called
\emph{doubly-even} if the weight of any codeword is divisible by 4
and \emph{singly-even} otherwise. By \cite{wood}, the ring $R_{k,m}$
is suitable to study self-dual codes;

\begin{lemma}
\cite{tufekci}A linear code $\mathcal{C}$ of length $n$ over
$R_{k,m}$ satisfies $\left\vert \mathcal{C}\right\vert .\left\vert
\mathcal{C}^{\perp }\right\vert =\left\vert R_{k,m}\right\vert
^{n}$.
\end{lemma}

\begin{definition}
\cite{tufekci}Take an element $\bar{a}=\overline{a}_{0}+\overline{a}_{1}u+%
\overline{a}_{2}u^{2}+\cdots +\overline{a}_{k-2}u^{k-2}+\overline{a}%
_{k-1}u^{k-1}$ of $(\mathcal{R}_{k,1})^{n}$, where
$\overline{a}_{i}\in
\mathbb{F}_{2}^{n}$. Then define the Gray map $\phi _{k1}$ from $(\mathcal{R}%
_{k,1})^{n}$ to $(\mathbb{F}_{2})^{kn}$ as follows: when $k$ is even
let
\begin{equation*}
\begin{tabular}{ll}
$\phi _{k1}(\bar{a})=$ & $(\overline{a}_{0}+\overline{a}_{1}+\cdots +%
\overline{a}_{k-2}+\overline{a}_{k-1},\overline{a}_{1}+\cdots +\overline{a}%
_{k-2}+\overline{a}_{k-1},$ \\
& $\overline{a}_{1}+\cdots +\overline{a}_{k-2},\cdots ,\overline{a}_{\frac{k%
}{2}-1}+\overline{a}_{\frac{k}{2}}+\overline{a}_{\frac{k}{2}+1},\overline{a}%
_{\frac{k}{2}-1}+\overline{a}_{\frac{k}{2}},\overline{a}_{\frac{k}{2}})$%
\end{tabular}%
\end{equation*}%
and when $k$ is odd let
\begin{equation*}
\begin{tabular}{ll}
$\phi _{k1}(\bar{a})=$ & $(\overline{a}_{0}+\overline{a}_{1}+\cdots +%
\overline{a}_{k-2}+\overline{a}_{k-1},\overline{a}_{1}+\cdots +\overline{a}%
_{k-2}+\overline{a}_{k-1},$ \\
& $\overline{a}_{1}+\cdots +\overline{a}_{k-2},\cdots ,\overline{a}_{\frac{%
k-3}{2}}+\overline{a}_{\frac{k-1}{2}}+\overline{a}_{\frac{k+1}{2}},\overline{%
a}_{\frac{k-1}{2}}+\overline{a}_{\frac{k+1}{2}},\overline{a}_{\frac{k-1}{2}%
}).$%
\end{tabular}%
\end{equation*}
\end{definition}

In \cite{tufekci}, the Gray map is extended to $R_{k,m}$ by viewing
$R_{k,m}$ as a vector space over $R_{k,1}$ basis $\left\{
1,v,v^{2},\ldots ,v^{m-1}\right\} $ as follows;
\begin{equation*}
\begin{tabular}{ll}
$\phi _{km}(c)=$ & ${\large {(}\phi _{k1}(\sum\limits_{i=0}^{m-1}\overline{c}%
_{ki}),\phi _{k1}(\sum\limits_{i=1}^{m-1}\overline{c}_{ki}),\phi
_{k1}(\sum\limits_{i=1}^{m-2}\overline{c}_{ki}),}$ \\
& $\cdots ,\phi _{k1}(\sum\limits_{i=\frac{m}{2}-1}^{\frac{m}{2}+1}\overline{%
c}_{ki}),\phi _{k1}(\sum\limits_{i=\frac{m}{2}-1}^{\frac{m}{2}}\overline{c}%
_{ki}),\phi _{k1}(\sum\limits_{i=\frac{m}{2}}^{\frac{m}{2}}\overline{c}_{ki})%
{\large {)}.}$%
\end{tabular}%
\end{equation*}%
where $c=\sum\limits_{0\leq i\leq m-1}c_{ki}v^{i}$ $c_{ki}\in
R_{k,1}$. The Lee weight $w_{L}$ of an element $a$ of $R_{k,m}$ is
defined to be the Hammimg weight of the Gray image. The Gray map
$\phi _{km}$ preserves duality.

\begin{theorem}
\cite{tufekci}\label{dual} Let $\mathcal{C}$ be a self-dual code over $%
R_{k,m}$ of length $n$. Then $\phi _{km}(\mathcal{C})$ is a binary
self-dual code of length $kmn$. Moreover the Lee weight distribution
of $C$ is the same as the Hamming weight distribution of $\phi
_{km}(\mathcal{C})$.
\end{theorem}

Consider the projections%
\begin{eqnarray*}
\pi _{v} &:&R_{k,m}\rightarrow R_{k,1}\text{ defined by }v\mapsto
0\text{,}
\\
\pi _{u} &:&R_{k,1}\rightarrow \mathbb{F}_{2}\text{ defined by }u\mapsto 0%
\text{,}
\end{eqnarray*}%
then $\mu =\pi _{u}\circ \pi _{v}$ is a projection from $R_{k,m}$ to $%
\mathbb{F}_{2}$. The projections preserve orthogonality and
projection of a free self-dual code is self-dual. The code
$\mathcal{D}$ is said to be a lift of $\mathcal{C}$ if its
projection is $\mathcal{C}$. The following theorem gives a bound for
the minimum distance of a lift;

\begin{theorem}
\cite{tufekci}\label{bound} Let $\mathcal{C}$ be a linear code over
$R_{k,m}$ of length $n$ with minimum Lee weight $d$ and $\mu
(\mathcal{C})$ be its projection to $\mathbb{F}_{2}$. If $d^{\prime
}$denotes the minimum Hamming weight of $\mu (\mathcal{C}),$ we have
$d\leq 2md^{\prime }.$
\end{theorem}

\section{Quadratic double circulant codes over $R_{k,m}$\label{qdcsection}}

Double circulant codes are a subfamily of quasi-cyclic codes. Double
circulant constructions are an effective method to form self-dual
codes. On the other hand, quadratic residue codes is another topic
of interest. In 2002, Gaborit defined quadratic double circulant
(QDC) codes as a generalization of quadratic residue codes in\
\cite{gaborit}. In this section, we study QDC codes over $R_{k,m}$.

Let $p$ be an odd prime and $Q_{p}\left( a,b,c\right) $ be the
circulant matrix\ with first row $r$ based on quadratic residues
modulo $p$ defined as $r\left[ 1\right] =a,$ $r\left[ i+1\right] =b$
if $i$ is a quadratic residue and $r\left[ i+1\right] =c$ if $i$ is
a quadratic non-residue modulo $p$. We state the special case of the
main theorem from \cite{gaborit} where $p$ is an odd prime;

\begin{theorem}
$($\cite{gaborit}\label{qdc}$)$ Let $p$ be an odd prime and let
$Q_{p}\left( a,b,c\right) $ be the circulant matrix with $a,b$ and
$c$ as the elements of
the ring $R_{k,m}$. If $p=4k+1$ then%
\begin{eqnarray}
&&Q_{p}\left( a,b,c\right) Q_{p}\left( a,b,c\right) ^{T}  \notag \\
&=&Q_{p}\left( a^{2}+2k\left( b^{2}+c^{2}\right) ,2ab-b^{2}+k\left(
b+c\right) ^{2},2ac-c^{2}+k\left( b+c\right) ^{2}\right) .
\label{mod41}
\end{eqnarray}%
If $p=4k+3$ then
\begin{eqnarray}
&&Q_{p}\left( a,b,c\right) Q_{p}\left( a,b,c\right) ^{T}  \notag \\
&=&Q_{p}(a^{2}+\left( 2k+1\right) \left( b^{2}+c^{2}\right)
,ab+ac+k\left(
b^{2}+c^{2}\right) +\left( 2k+1\right) bc,  \label{mod43} \\
&&ab+ac+k\left( b^{2}+c^{2}\right) +\left( 2k+1\right) bc).  \notag
\end{eqnarray}
\end{theorem}

\begin{definition}
$($\cite{gaborit}$)$ The code generated by $P_{p}\left( a,b,c\right)
=\left(
\begin{array}{c|c}
I_{p} & Q_{p}\left( a,b,c\right)
\end{array}%
\right) $ over $R_{k,m}$ is called a quadratic double circulant code
and is denoted by $\mathcal{QDC}_{p}\left( R_{k,m}\right) \left(
a,b,c\right) $.
\end{definition}

\begin{example}
Consider the code $\mathcal{QDC}_{p}\left( R_{2,2}\right) \left(
1+v+uv,u,v\right) $ that is generated by%
\begin{equation*}
\left[ \enspace I_{5}\enspace%
\begin{array}{|ccccc}
1+v+uv & u & v & v & u \\
u & 1+v+uv & u & v & v \\
v & u & 1+v+uv & u & v \\
v & v & u & 1+v+uv & u \\
u & v & v & u & 1+v+uv%
\end{array}%
\right] .
\end{equation*}%
Self-duality of the code is easily checked by Theorem \ref{qdc}.
Moreover, each row of the generator matrix has Lee weight $8$, which
means the binary
image of the code is doubly-even. It is an extremal self-dual $\left[ 40,20,8%
\right] $ code with partial weight distribution $1+285z^{8}+21280z^{12}+%
\cdots $.
\end{example}

In the following, we define a special subfamily of units and non-units in $%
R_{k,m}$;

\begin{definition}
\label{basic}An element $r$ of $R_{k,m}$ is called a \emph{basic
non-unit} if $r^{2}=0$ and a \emph{basic unit} if $r^{2}=1$.
\end{definition}

It is easily observed that $1+r$ is a basic unit if and only if $r$
is a basic non-unit.

In the following theorems families of self-dual QDC codes over
$R_{k,m}$ are given.

\begin{theorem}
Let $a$ be an element of $R_{k,m}$ such that $a^{3}=0$ and $p$ be a
prime with $p\equiv 3\pmod{8}$ then the codes
\begin{equation*}
\mathcal{QDC}_{p}\left( R_{k,m}\right) \left( a,1,a+a^{2}\right) \text{ and }%
\mathcal{QDC}_{p}\left( R_{k,m}\right) \left(
a,1+a^{2},a+a^{2}\right)
\end{equation*}
are self-dual. The constructions are called $I$ and $II$,
respectively.
\end{theorem}

\begin{proof}
Since $p=8k+3$, $a^{3}=0$ and $char\left( R_{k,m}\right) =2,$ by the
equation \ref{mod43} we have
\begin{eqnarray*}
&&Q_{p}\left( a,1,a+a^{2}\right) Q_{p}\left( a,1,a+a^{2}\right) ^{T} \\
&=&Q_{p}\left( a^{2}+1+\left( a+a^{2}\right) ^{2},a+a\left(
a+a^{2}\right) +\left( a+a^{2}\right) ,a+a\left( a+a^{2}\right)
+\left( a+a^{2}\right)
\right) \\
&=&Q_{p}\left( 0,1,1\right) =I_{p},
\end{eqnarray*}%
which implies that $\mathcal{QDC}_{p}\left( R_{k,m}\right) \left(
a,1,a+a^{2}\right) $ is self-dual. By analogous steps $\mathcal{QDC}%
_{p}\left( R_{k,m}\right) \left( a,1+a^{2},a+a^{2}\right) $ is also
self-dual.
\end{proof}

The characterization of non-units given in Definition \ref{basic}
can be used to construct self-dual codes as follows;

\begin{theorem}
Let $a$ and $b$ be two basic non-units in $R_{k,m}$ and $p$ be a
prime then the code $\mathcal{QDC}_{p}\left( R_{k,m}\right) \left(
1+a,a,b\right) $ is self-dual whenever $p\equiv 1\pmod{4}$.
Moreover, $\mathcal{QDC}_{p}\left(
R_{k,m}\right) \left( a,1+b,a\right) $ is self-dual if $ab=0$ and $p\equiv 3%
\pmod{8}$. The constructions are called as $III$ and $IV$,
respectively.
\end{theorem}

\begin{proof}
Let $p=4k+1$ be a prime, $a$ and $b$ be basic non-units in $R_{k,m}$
then by equation \ref{mod41}
\begin{eqnarray*}
&&Q_{p}\left( 1+a,a,b\right) Q_{p}\left( 1+a,a,b\right) ^{T} \\
&=&\left\{
\begin{array}{c}
Q_{p}\left( \left( 1+a\right) ^{2},a^{2},b^{2}\right) \text{ if
}k\text{ is
even} \\
Q_{p}\left( \left( 1+a\right) ^{2},b^{2},a^{2}\right) \text{ if
}k\text{ is
odd}%
\end{array}%
\right.  \\
&=&Q_{p}\left( 1,0,0\right) =I_{p}.
\end{eqnarray*}%
Hence, the code $\mathcal{QDC}_{p}\left( R_{k,m}\right) \left(
1+a,a,b\right) $ is self-dual.

Let $p=8k+3$, $a$ and $b$ be basic non-units in $R_{k,m}$ with
$ab=0$. Then,
since $char\left( R_{k,m}\right) =2,$ by equation \ref{mod43} we have%
\begin{eqnarray*}
&&Q_{p}\left( a,1+b,a\right) Q_{p}\left( a,1+b,a\right) ^{T} \\
&=&Q_{p}\left( 1,a\left( 1+b\right) +\left( 1+b\right) a,a\left(
1+b\right)
+\left( 1+b\right) a\right) \\
&=&Q_{p}\left( 1,0,0\right) =I_{p}.
\end{eqnarray*}%
Therefore, the code $\mathcal{QDC}_{p}\left( R_{k,m}\right) \left(
a,1+b,a\right) $ is self-dual.
\end{proof}
We list some good QDC codes over $R_{k,m}$ in Table \ref{tab:qdc}.
\begin{table}[H]
\caption{Some examples of self-dual QDC codes over $R_{k,m}$}
\label{tab:qdc}\centering
\begin{tabular}{||c|c|c|c|c|l||}
\hline $R$ & $p$ & Construction & $a,\left( b\right) $ & The binary
image & Comment
\\ \hline\hline
$R_{2,1}$ & $5$ & $III$ & $u,0$ & $\left[ 20,10,4\right] $ & extremal \\
\hline $R_{2,2}$ & $5$ & $III$ & $u,v$ & $\left[ 40,20,8\right] $ &
extremal singly-even \\ \hline $R_{2,2}$ & $5$ & $III$ & $u+uv,v$ &
$\left[ 40,20,8\right] $ & extremal doubly-even \\ \hline
$R_{2,1}$ & $11$ & $I,$ $II$ & $u$ & $\left[ 44,22,8\right] $ & extremal \\
\hline $R_{3,1}$ & $11$ & $I$ & $u$ & $\left[ 66,33,12\right] $ &
extremal \\ \hline
$R_{3,1}$ & $11$ & $II$ & $u$ & $\left[ 66,33,12\right] $ & extremal \\
\hline
$R_{2,2}$ & $11$ & $II$ & $uv$ & $\left[ 88,44,12\right] $ & singly-even \\
\hline $R_{2,2}$ & $11$ & $IV$ & $u,uv$ & $\left[ 88,44,12\right] $
& doubly-even
\\ \hline
$R_{4,1}$ & $11$ & $I$ & $u^{3}$ & $\left[ 88,44,12\right] $ &
singly-even
\\ \hline
$R_{3,1}$ & $19$ & $I$ & $u$ & $\left[ 114,57,16\right] $ & - \\
\hline $R_{3,2}$ & $11$ & $II$ & $v+uv$ & $\left[ 132,66,12\right] $
& - \\ \hline $R_{4,1}$ & $19$ & $I$ & $u^{3}$ & $\left[
152,76,16\right] $ & singly-even
\\ \hline
\end{tabular}%
\end{table}

\section{Constructions for self-dual codes over $R_{k,m}$ by $\protect%
\lambda $-circulant matrices\label{64section}}

In this section, the four circulant construction is generalized to $\lambda $%
-circulant matrices. Extremal singly-even binary self-dual codes of length $%
64$ are constructed as Gray images of four circulant codes over
$\mathbb{F}_{2}$ and $R_{2,2}$. The codes are going to be used in
Section \ref{newcodes} to construct new binary self-dual codes of
lengths $66$ and $68$.

The possible weight enumerators of singly-even extremal self-dual
codes of length $64$ are characterized in \cite{conway} as:
\begin{eqnarray*}
W_{64,1} &=&1+(1312+16\beta )y^{12}+(22016-64\beta )y^{14}+\cdots
\text{
where }14\leq \beta \leq 104, \\
W_{64,2} &=&1+(1312+16\beta )y^{12}+(23040-64\beta )y^{14}+\cdots
\text{ where }0\leq \beta \leq 277\text{.}
\end{eqnarray*}
Recently, codes with $\beta =$29, 39, 53 and 60 in $W_{64,1}$ and
codes with $\beta =$51, 58 in $W_{64,2}$ are constructed in
\cite{yankov} and a code with $\beta =80$ in $W_{64,2}$ is
constructed in \cite{karadeniz}. Together with these the existence
of such codes is now known for $\beta =$14, 18, 22, 25, 29, 32, 36,
39, 44, 46, 53, 60, 64 in $W_{64,1}$ and for $\beta =$0, 1, 2, 4, 5,
6, 8, 9, 10, 12, 13, 14, 16, 17, 18, 20, 21, 22, 23, 24, 25, 28, 29,
30, 32, 33, 36, 37, 38, 40, 41, 44, 48, 51, 52, 56, 58, 64, 72, 80,
88, 96, 104, 108, 112, 114, 118, 120, 184 in $W_{64,2}$.

The four circulant construction was defined in \cite{betsumiya}.
\begin{definition}
Let $r=\left( r_{1},r_{2},\ldots ,r_{n}\right) $ be an element of
$\left(
R_{k,m}\right) ^{n}$. The $\lambda $-cyclic shift of $r$ is defined as $%
\sigma _{\lambda }\left( r\right) =\left( \lambda
r_{n},r_{1},r_{2},\ldots
,r_{m-1}\right) $ where $\lambda \in R$. A square matrix is called $\lambda $%
-circulant if every row is the $\lambda $-cyclic shift of the
previous one.
\end{definition}
Since $\lambda $-circulant matrices commute with each other the four
circulant construction can be extended to $\lambda $-circulant
matrices. We have the following result:

\begin{theorem}
\label{fourcirc} Let $\mathcal{C}$ be the linear code over $R_{k,m}$
of length $4n$ generated by the four circulant matrix
\begin{equation*}
G:=\left[ \enspace I_{2n}\enspace%
\begin{array}{|cc}
A & B \\
B^{T} & A^{T}%
\end{array}%
\right]
\end{equation*}%
where $A$ and $B$ are $\lambda $-circulant $n\times n$ matrices over $%
R_{k,m} $ satisfying $AA^{T}+BB^{T}=I_{n}$. Then the code
$\mathcal{C}$ is called a $\lambda $-four circulant code over
$R_{k,m}$. The code $\mathcal{C} $ and its binary image are
self-dual.
\end{theorem}

Four circulant codes of length $32$ over $R_{2,1}$ have been studied
extensively in \cite{karadeniz} and the codes with weight
enumerators $\beta =0$, $16$, $32$, $48$ and $80$ in $W_{64,2}$ were
obtained. The code with the weight enumerator $\beta =80$ in
$W_{64,2}$ is the first such code in literature. For further
reference we name this code as $\mathcal{C}_{64,80}$ which is the
four circulant code over $R_{2,1}$ with
\begin{equation*}
r_{A}=\left( u,0,0,0,u,1,u,1+u\right) \text{ and }r_{B}=\left(
u,u,0,1,1,1+u,1+u,1+u\right) \text{.}
\end{equation*}

By considering $\left( 1+u\right) $-four circulant codes of length
$32$ over $R_{2,1}$ we were able to obtain the binary codes with
weight enumerators for $\beta =8k$ in $W_{64,2}$ where $0\leq k\leq
9$. These are listed in Table \ref{tab:R1}.
\begin{table}[H]
\caption{$\left( 1+u\right) $-four circulant codes over $R_{2,1}$}
\label{tab:R1}\centering
\begin{tabular}{||l|l|l|c|c||}
\hline $\mathcal{L}_{i}$ & $r_{A}$ & $r_{B}$ & $\beta $ in
$W_{64,2}$ & $\left\vert Aut\left( \mathcal{L}_{i}\right)
\right\vert $ \\ \hline\hline
$\mathcal{L}_{1}$ & $\left( u333uuu0\right) $ & $\left( 11311010\right) $ & $%
8$ & $2^{5}$ \\ \hline
$\mathcal{L}_{2}$ & $\left( u111000u\right) $ & $\left( 11333u1u\right) $ & $%
24$ & $2^{5}$ \\ \hline
$\mathcal{L}_{3}$ & $\left( u131u0uu\right) $ & $\left( 31313030\right) $ & $%
72$ & $2^{5}$ \\ \hline
$\mathcal{L}_{4}$ & $\left( 33uu3110\right) $ & $\left( 113u00u3\right) $ & $%
0$ & $2^{5}$ \\ \hline
$\mathcal{L}_{5}$ & $\left( 330u3110\right) $ & $\left( 1310uuu1\right) $ & $%
16$ & $2^{5}$ \\ \hline
$\mathcal{L}_{6}$ & $\left( 33uu3130\right) $ & $\left( 331u0u01\right) $ & $%
32$ & $2^{5}$ \\ \hline
$\mathcal{L}_{7}$ & $\left( 11u03130\right) $ & $\left( 131u0003\right) $ & $%
48$ & $2^{5}$ \\ \hline
$\mathcal{L}_{8}$ & $\left( 310u113u\right) $ & $\left( 1330uu03\right) $ & $%
64$ & $2^{6}$ \\ \hline
$\mathcal{L}_{9}$ & $\left( u1110u3u\right) $ & $\left( 30u03113\right) $ & $%
8$ & $2^{5}$ \\ \hline $\mathcal{L}_{10}$ & $\left( 0133uu30\right)
$ & $\left( 10001113\right) $ & $24$ & $2^{5}$ \\ \hline
$\mathcal{L}_{11}$ & $\left( u111001u\right) $ & $\left(
3u0u1311\right) $ & $40$ & $2^{5}$ \\ \hline $\mathcal{L}_{12}$ &
$\left( 0133u01u\right) $ & $\left( 10001311\right) $ & $56$ &
$2^{5}$ \\ \hline
\end{tabular}%
\end{table}
In order to construct extremal binary self-dual codes of length $64$
as Gray images of $\lambda $-four circulant codes of length $16$
over $R_{2,2}$ we
lift binary codes to codes over $R_{2,1}$ and then lift these to codes over $%
R_{2,2}$. Theorem \ref{bound} tells us the minimum distance of the
codes to be lifted. We demonstrate this in the following example;

\begin{example}
Let $\mathcal{C}$ be the four circulant code of length $16$ over $\mathbb{F}%
_{2}$ with $r_{A}=\left( 1,0,0,0\right) $ and $r_{B}=\left( 1,1,1,1\right) $%
. Then $\mathcal{C}$ is a singly-even $\left[ 16,8,4\right] ~$code.
The code $\mathcal{C}$ is lifted to $\mathcal{C}^{\prime }$, which
is the $\left(
1+u\right) $-four circulant code of length $16$ over $R_{2,1}$ with $%
r_{A}^{\prime }=\left( 1,0,u,u\right) $ and $r_{B}^{\prime }=\left(
1,1+u,1,1+u\right) $. The binary image $\phi
_{21}(\mathcal{C}^{\prime })$ of $\mathcal{C}^{\prime }$ is a
self-dual $\left[ 32,16,6\right] $ code. Then $\mathcal{C}^{\prime
}$ is lifted to the $\mathcal{C}^{\prime \prime }$ that is the
$\left( 1+u+v+uv\right) $-four circulant code of length $16$ over
$R_{2,2}$ with
\begin{equation*}
r_{A}^{\prime \prime }=\left( 1,0,u,u+v+uv\right) \text{ and
}r_{B}^{\prime \prime }=\left( 1+v+uv,1+u,1+v,1+u+v\right) .
\end{equation*}%
The binary code $\phi _{22}(\mathcal{C}^{\prime \prime })$ is and
extremal
singly-even binary self-dual code of length $64$ with weight enumerator $%
\beta =0$ in $W_{64,2}$. Note that, $\pi _{v}\left(
\mathcal{C}^{\prime \prime }\right) =\mathcal{C}^{\prime }$, $\pi
_{u}\left( \mathcal{C}^{\prime
}\right) =\mathcal{C}$ and $\mu \left( \mathcal{C}^{\prime \prime }\right) =%
\mathcal{C}$.
\end{example}

In order to fit the upcoming tables we use hexadecimal number sytem.
The one-to-one correspondence between hexadecimals and binary $4$
tuples is as follows:
\begin{eqnarray*}
0 &\leftrightarrow &0000,\ 1\leftrightarrow 0001,\ 2\leftrightarrow
0010,\
3\leftrightarrow 0011, \\
4 &\leftrightarrow &0100,\ 5\leftrightarrow 0101,\ 6\leftrightarrow
0110,\
7\leftrightarrow 0111, \\
8 &\leftrightarrow &1000,\ 9\leftrightarrow 1001,\ A\leftrightarrow
1010,\
B\leftrightarrow 1011, \\
C &\leftrightarrow &1100,\ D\leftrightarrow 1101,\ E\leftrightarrow
1110,\ F\leftrightarrow 1111.
\end{eqnarray*}

To express elements of $R_{2,2}$ we use the ordered basis $\left\{
uv,v,u,1\right\} $. For instance $1+u+uv$ in $R_{2,2}$ is expressed
as $1011$ which is $B$. By considering $\lambda $-four circulant
codes of length $16$
over $R_{2,2}$ we obtain self-dual binary codes with weight enumerators in $%
W_{64,2}$ for various values for $\beta $, these are listed in Table \ref%
{tab:R2}.

\begin{table}[H]
\caption{Self-dual $\protect\lambda $-four circulant codes over
$R_{2,2}$} \label{tab:R2}\centering
\begin{tabular}{||l|l|l|l|c|c||}
\hline $\mathcal{M}_{i}$ & $\lambda $ & $r_{A}$ & $r_{B}$ & $\beta $
in $W_{64,2}$ & $\left\vert Aut\left( \mathcal{M}_{i}\right)
\right\vert $ \\ \hline\hline $\mathcal{M}_{1}$ & $3$ & $\left(
F,0,E,2\right) $ & $\left( 7,5,3,D\right) $ & $0$ & $2^{5}$ \\
\hline $\mathcal{M}_{2}$ & $3$ & $\left( 7,0,C,A\right) $ & $\left(
F,F,9,5\right) $ & $16$ & $2^{5}$ \\ \hline $\mathcal{M}_{3}$ & $3$
& $\left( 3,0,D,4\right) $ & $\left( E,3,F,B\right) $ & $48$ &
$2^{5}$ \\ \hline $\mathcal{M}_{4}$ & $7$ & $\left( B,0,1,C\right) $
& $\left( 9,B,1,2\right) $ & $5$ & $2^{3}$ \\ \hline
$\mathcal{M}_{5}$ & $7$ & $\left( B,0,1,4\right) $ & $\left(
A,7,5,F\right) $ & $8$ & $2^{4}$ \\ \hline $\mathcal{M}_{6}$ & $7$ &
$\left( 3,0,7,A\right) $ & $\left( B,C,D,9\right) $ & $9$ & $2^{3}$
\\ \hline $\mathcal{M}_{7}$ & $7$ & $\left( 7,0,5,C\right) $ &
$\left( 1,3,2,5\right) $ & $12$ & $2^{4}$ \\ \hline
$\mathcal{M}_{8}$ & $7$ & $\left( D,0,F,C\right) $ & $\left(
F,1,7,A\right) $ & $13$ & $2^{3}$ \\ \hline $\mathcal{M}_{9}$ & $7$
& $\left( B,0,1,C\right) $ & $\left( A,5,5,D\right) $ & $16$ &
$2^{5}$ \\ \hline $\mathcal{M}_{10}$ & $7$ & $\left( B,0,F,A\right)
$ & $\left( B,C,D,7\right) $ & $17$ & $2^{3}$ \\ \hline
$\mathcal{M}_{11}$ & $7$ & $\left( 7,0,5,C\right) $ & $\left(
2,7,5,F\right) $ & $24$ & $2^{4}$ \\ \hline $\mathcal{M}_{12}$ & $F$
& $\left( 1,0,2,E\right) $ & $\left( D,3,5,7\right) $ & $0$ &
$2^{5}$ \\ \hline $\mathcal{M}_{13}$ & $F$ & $\left( C,0,3,6\right)
$ & $\left( 1,B,7,1\right) $ & $16$ & $2^{5}$ \\ \hline
$\mathcal{M}_{14}$ & $F$ & $\left( F,0,B,A\right) $ & $\left(
F,B,4,5\right) $ & $48$ & $2^{5}$ \\ \hline $\mathcal{M}_{15}$ & $B$
& $\left( 9,0,F,C\right) $ & $\left( B,6,9,3\right) $ & $5$ &
$2^{3}$ \\ \hline $\mathcal{M}_{16}$ & $B$ & $\left( D,0,3,C\right)
$ & $\left( 6,B,5,3\right) $ & $8$ & $2^{4}$ \\ \hline
$\mathcal{M}_{17}$ & $B$ & $\left( 5,0,B,4\right) $ & $\left(
7,6,D,9\right) $ & $9$ & $2^{3}$ \\ \hline $\mathcal{M}_{18}$ & $B$
& $\left( 5,0,1,E\right) $ & $\left( 9,9,C,B\right) $ & $12$ &
$2^{4}$ \\ \hline $\mathcal{M}_{19}$ & $B$ & $\left( D,0,1,6\right)
$ & $\left( F,1,7,C\right) $ & $13$ & $2^{3}$ \\ \hline
$\mathcal{M}_{20}$ & $B$ & $\left( 5,0,B,C\right) $ & $\left(
E,D,F,5\right) $ & $16$ & $2^{3}$ \\ \hline $\mathcal{M}_{21}$ & $B$
& $\left( B,0,5,C\right) $ & $\left( 7,E,D,7\right) $ & $17$ &
$2^{3}$ \\ \hline $\mathcal{M}_{22}$ & $B$ & $\left( D,0,3,4\right)
$ & $\left( E,9,3,1\right) $ & $24$ & $2^{4}$ \\ \hline
\end{tabular}%
\end{table}

\begin{remark}
In order to construct the codes in Table \ref{tab:R1} the binary
four circulant codes are lifted to $R_{2,1}$. Similarly, to
construct the codes in Table \ref{tab:R2} the binary four circulant
codes are lited to $R_{2,1}$
and then to $R_{2,2}$. This reduces the search field remarkably from $%
2^{32}=4294967296$ to $2^{16}=65536$.
\end{remark}

\section{New binary self-dual codes by extensions\label{newcodes}}

By applying the extension theorems to the self-dual codes
constructed in Section \ref{64section} we were able to obtain new
binary self-dual codes of lengths $66$ and $68$. In particular we
were able to construct $11$ new codes of length $66$ and $34$ new
codes of length $68$. Extensions for self-dual codes were first used
by Brualdi and Pless in \cite{brualdi}. Since then different
versions of extensions applied, for some of these we refer to
\cite{kim, doughertyfrobenius} and \cite{kayayildiz}. The following
extension theorems hold for any commutative Frobenius ring $R$ of
characteristic $2$.

\begin{theorem}
$($\cite{doughertyfrobenius}$)$ \label{ext}Let $\mathcal{C}$ be a
self-dual code over $R$ of length $n$ and $G=(r_{i})$ be a $k\times
n$ generator matrix for $\mathcal{C}$, where $r_{i}$ is the $i$-th
row of $G$, $1\leq i\leq k$. Let $c$ be a unit in $R$ such that
$c^{2}=1$ and $X$ be a vector in $R^{n}$ with $\left\langle
X,X\right\rangle =1$. Let $y_{i}=\left\langle
r_{i},X\right\rangle $ for $1\leq i\leq k$. Then the following matrix%
\begin{equation*}
\left(
\begin{array}{cc|c}
1 & 0 & X \\ \hline
y_{1} & cy_{1} & r_{1} \\
\vdots & \vdots & \vdots \\
y_{k} & cy_{k} & r_{k}%
\end{array}%
\right) ,
\end{equation*}%
generates a self-dual code $\mathcal{C}^{\prime }$ over $R$ of
length $n+2$.
\end{theorem}

A more specific extension method which can be applied to generator
matrices in standard form is as follows:

\begin{theorem}
$($\cite{kayayildiz}$)$ \label{idext}Let $\mathcal{C}$ be a
self-dual code generated by $G=\left( I_{n}|A\right) $ over $R$. If
the sum of the elements
in $i$-th row of $A$ is $r_{i}$ then the matrix:%
\begin{equation*}
G^{\ast }=\left(
\begin{array}{cc|cccccc}
1 & 0 & x_{1} & \ldots & x_{n} & 1 & \ldots & 1 \\ \hline
y_{1} & cy_{1} & \multicolumn{3}{c}{} & \multicolumn{3}{c}{} \\
\vdots & \vdots &  & I_{n} &  &  & A &  \\
y_{n} & cy_{n} &  &  &  &  &  &
\end{array}%
\right) ,
\end{equation*}%
where $y_{i}=x_{i}+r_{i}$, $c$ is a unit with $c^{2}=1$, $X=\left(
x_{1},\ldots ,x_{n}\right) $ and $\left\langle X,X\right\rangle
=1+n$, generates a self-dual code $\mathcal{C}^{\ast }$ over $R$.
\end{theorem}

\subsection{$\mathbb{F}_{2}$-extensions}

The Gray images of the codes in tables \ref{tab:R1} and \ref{tab:R2}
are extremal singly-even self-dual binary codes of length $64$. In
this section, we construct extremal binary self-dual codes of length
$66$ by applying Theorem \ref{ext}. Eleven new codes are obtained.

We recall that a self-dual $\left[ 66,33,12\right] _{2}$-code has a
weight enumerator in
one of the following forms \cite{dougherty};%
\begin{eqnarray*}
W_{66,1} &=&1+\left( 858+8\beta \right) y^{12}+\left( 18678-24\beta
\right)
y^{14}+\cdots \text{ where }0\leq \beta \leq 778, \\
W_{66,2} &=&1+1690y^{12}+7990y^{14}+\cdots \text{ } \\
\text{and }W_{66,3} &=&1+\left( 858+8\beta \right) y^{12}+\left(
18166-24\beta \right) y^{14}+\cdots \text{ where }14\leq \beta \leq
756,
\end{eqnarray*}%
Recently, five new codes in $W_{66,1}$ are constructed in
\cite{karadeniz}. The existence of such codes is known for $\beta
=$0, 1, 2, 3, 5, 6, 8,. . . , 11, 14,. . . ,18, 20,... , 54, 56, 59,
60, 62,... 69, 71,. . . , 74, 76, 77, 78, 80, 83, 84, 86, 87,
92, 94 in $W_{66,1}$. For a list of known codes in $W_{66,3}$ we refer to \cite%
{karadeniz66}.

We construct the codes with weight enumerators $\beta =$19, 61, 75,
79, 81, 82, 85, 88, 89, 90 and 100 in $W_{66,1}$. The extension in
Theorem \ref{ext}
is applied to the binary images of the codes constructed in Section \ref%
{64section} to obtain the new codes. The results are given in Table \ref%
{tab:66} where $\boldsymbol{1}^{\boldsymbol{32}}$ denotes 32
successive 1s in $X$.
\begin{table}[H]
\caption{New extremal binary self-dual codes with weight enumerators in $%
W_{66,1}$ by Theorem \protect\ref{ext} (11 codes)}
\label{tab:66}\centering
\begin{tabular}{||c|c|c||}
\hline
Code & The extension vector $X$ & $\beta $ in $W_{66,1}$ \\
\hline\hline
$\mathcal{M}_{17}$ & $10101010111001110010001101110010\boldsymbol{1}^{%
\boldsymbol{32}}$ & 19 \\ \hline
$\mathcal{M}_{3}$ & $11001100001011100111100101011111\boldsymbol{1}^{%
\boldsymbol{32}}$ & 61 \\ \hline
$\mathcal{L}_{8}$ & $10001011111111011011010110100100\boldsymbol{1}^{%
\boldsymbol{32}}$ & 75 \\ \hline $\mathcal{L}_{8}$ &
\begin{tabular}{l}
$00010101001111110101110111100101$ \\
$01111001110010000111111001100000$%
\end{tabular}
& 79 \\ \hline $\mathcal{L}_{3}$ &
\begin{tabular}{l}
$01100110100001001100000110100000$ \\
$01001101100110110111110101111001$%
\end{tabular}
& 81 \\ \hline $\mathcal{L}_{8}$ &
\begin{tabular}{l}
$01010110111110101100011010100111$ \\
$00010101100101110100110101101001$%
\end{tabular}
& 82 \\ \hline $\mathcal{L}_{3}$ &
\begin{tabular}{l}
$00111101100000000111010010101001$ \\
$00100001110000111110001100010100$%
\end{tabular}
& 85 \\ \hline
$\mathcal{C}_{64,80}$ & $11100000101011010111100100110110\boldsymbol{1}^{%
\boldsymbol{32}}$ & 88 \\ \hline
$\mathcal{C}_{64,80}$ & $10100100001110101110100111000001\boldsymbol{1}^{%
\boldsymbol{32}}$ & 89 \\ \hline
$\mathcal{C}_{64,80}$ & $00011111110111101111001110001011\boldsymbol{1}^{%
\boldsymbol{32}}$ & 90 \\ \hline
$\mathcal{C}_{64,80}$ & $11100001100000000001000010011011\boldsymbol{1}^{%
\boldsymbol{32}}$ & 100 \\ \hline
\end{tabular}%
\end{table}

\subsection{$R_{2,1}$-extensions}

In this section, we obtain new extremal binary self-dual codes of
length $68$ by considering $R_{2,1}$-extensions of the codes
constructed in the previous section. The ring $R_{2,2}$ can be
considered as an extension of $R_{2,1}$.
Throughout this section, $\varphi _{u}$ is the Gray map from $R_{2,2}$ to $%
R_{2,1}$ defined as $\varphi _{u}\left( a+bv\right) =\left(
b,a+b\right) $
where $a,b\in R_{2,1}$. We consider the extensions of the codes in Table \ref%
{tab:R1} as well as the Gray images of the codes in Table
\ref{tab:R2} under $\varphi _{u}$. 39 new extremal binary self-dual
codes of length $68$ are obtained as the binary images of the
extensions.

The weight enumerator of an extremal binary self-dual code of length
$68$ is characterized in \cite{dougherty} as follows:
\begin{eqnarray*}
W_{68,1} &=&1+\left( 442+4\beta \right) y^{12}+\left( 10864-8\beta
\right)
y^{14}+\cdots \text{ , }104\leq \beta \leq 1358\text{,} \\
W_{68,2} &=&1+\left( 442+4\beta \right) y^{12}+\left( 14960-8\beta
-256\gamma \right) y^{14}+\cdots
\end{eqnarray*}%
where $0\leq \gamma \leq 11$ and $14\gamma \leq \beta \leq
1870-32\gamma $. Tsai et al. constructed new extremal self-dual
binary codes of lengths $66$ and $68$ in \cite{tsai}. Recently, $3$
codes with previously unknown weight enumerators in $W_{68,1}$ were
constructed in \cite{jacodes}. Together with
the codes obtained in \cite{tsai, jacodes} the existence of codes in $%
W_{68,1}$ are known for $\beta =$104, 117, 120, 122, 123, 125,\ldots
,168, 170,\ldots ,232, 234, 235, 236, 241, 255, 257,\ldots ,269,
302, 328,\ldots , 336, 338, 339, 345, 347, 355, 401.

We obtain a code with a weight enumerator $\beta =169$ in
$W_{68,1}$.

First codes with $\gamma =4$ and $\gamma =6$ in $W_{68,2}$ are
constructed in \cite{karadeniz68}. Recently, new codes in $W_{68,2}$
are obtained in \cite{kayayildiz,kaya,jacodes} together with these,
codes exist for $W_{68,2} $ when
\begin{eqnarray*}
\gamma  &=&0,\ \beta =44,...,154\text{ or }\beta \in \left\{
2m|m=\text{19,
20, 88, 102, 119, 136 or }78\leq m\leq 86\right\} ; \\
\gamma  &=&1,\ \beta =49,57,59,...,160\text{ or }\beta \in \left\{ 2m|m=%
\text{27, 28, 29, 95, 96 or }81\leq m\leq 89\right\} ; \\
\gamma  &=&2,\ \beta =65,68,69,71,77,81,159\text{ or }\beta \in
\left\{
2m|37\leq m\leq 68,\text{ }70\leq m\leq 81\right\} \text{ or} \\
\beta  &\in &\left\{ 2m+1|42\leq m\leq 69,\text{ }71\leq m\leq
77\right\} ;
\\
\gamma  &=&3,\ \beta =101,117,123,127,133,137,141,145,147,149,153,159,193%
\text{ or } \\
\beta  &\in &\left\{ 2m|m=\text{%
44,45,48,50,51,52,54,...,58,61,63,...,66,68,...,72,74,77,...,81,88,94,98}%
\right\} ; \\
\gamma  &=&4\text{, }\beta \in \left\{ 2m|m=\text{51, 55, 58, 60,
61, 62,
64, 65, 67,...,71, 75,..., 78, 80}\right\} \text{ and } \\
\gamma  &=&6\text{ with }\beta \in \left\{ 2m|m=\text{69, 77, 78, 79, 81, 88}%
\right\} \text{.}
\end{eqnarray*}

In this section, we construct the codes with weight enumerators in
$W_{68,2}$ for $\gamma =0$ and $\beta =178$; $\gamma =1$ and $\beta
=180$; $\gamma =2$ and $\beta =$60, 62, 64, 66, 70, 72, 164, 166,
168, 170, 172, 174, 176, 178, 180, 182, 186; $\gamma =3$ and $\beta
=$94, 107, 118, 120, 156, 168, 172, 180; $\gamma =4$ and $\beta
=$98, 104, 108, 112, 174, 194.

By considering $R_{2,1}$-extensions of codes in Table \ref{tab:R1}
with respect to Theorem \ref{idext} we were able to obtain 14 new
extremal binary self-dual codes, which are listed in Table
\ref{tab:idext}.

\begin{table}[H]
\caption{New codes in $W_{68,2}$ by Theorem \protect\ref{idext} on
$R_{2,1}$ (14 codes)} \label{tab:idext}\centering
\begin{tabular}{||l|l|l|l|l||}
\hline $\mathcal{L}_{i}$ & $X$ & $c$ & $\gamma $ & $\beta $ \\
\hline\hline $\mathcal{L}_{4}$ & $\left( 1313uu0133130u11\right) $ &
$1+u$ & $2$ & $60$
\\ \hline
$\mathcal{L}_{4}$ & $\left( 1131uu011133u011\right) $ & $1$ & $2$ & $62$ \\
\hline
$\mathcal{L}_{4}$ & $\left( 0001u11uu3110300\right) $ & $1$ & $2$ & $64$ \\
\hline $\mathcal{L}_{4}$ & $\left( 00u1u130u111u1u0\right) $ & $1+u$
& $2$ & $66$
\\ \hline
$\mathcal{L}_{4}$ & $\left( uuu30330013101uu\right) $ & $1+u$ & $2$
& $70$
\\ \hline
$\mathcal{L}_{4}$ & $\left( u0u1u13uu333u3u0\right) $ & $1+u$ & $2$
& $72$
\\ \hline
$\mathcal{L}_{3}$ & $\left( u3000uu33u31u031\right) $ & $1$ & $2$ & $166$ \\
\hline $\mathcal{L}_{3}$ & $\left( u1u0u0u11u31uu13\right) $ & $1+u$
& $2$ & $170$
\\ \hline
$\mathcal{L}_{3}$ & $\left( 03u0u00330310u31\right) $ & $1+u$ & $2$
& $172$
\\ \hline
$\mathcal{L}_{3}$ & $\left( u1uuu0u11u31u013\right) $ & $1+u$ & $2$
& $174$
\\ \hline
$\mathcal{L}_{3}$ & $\left( 01000u0110310013\right) $ & $1+u$ & $2$
& $176$
\\ \hline
$\mathcal{L}_{3}$ & $\left( 011300u031111313\right) $ & $1$ & $3$ & $156$ \\
\hline $\mathcal{L}_{3}$ & $\left( 3u131011301u0u10\right) $ & $1+u$
& $3$ & $172$
\\ \hline
$\mathcal{L}_{3}$ & $\left( 103130333010u010\right) $ & $1+u$ & $3$
& $180$
\\ \hline
\end{tabular}%
\end{table}

\begin{example}
Let $\mathcal{C}$ be the code obtained by applying Theorem \ref{ext} for $%
\mathcal{\varphi }_{u}\left( M_{4}\right) $ over $R_{2,1}$ with%
\begin{equation*}
X=\left(
u,1+u,0,0,0,1+u,0,0,1,u,0,1,u,u,1+u,0,1111111111111111\right)
\end{equation*}%
and $c=1+u$ then the binary image of the extension is an extremal
binary
self-dual code of length $68$ with a weight enumerator $\beta =169$ in $%
W_{68,1}$. The code $\mathcal{C}$ is the first extremal binary
self-dual code with this weight enumerator.
\end{example}

Theorem \ref{ext} is applied to codes in Table \ref{tab:R1} and $R_{2,1}$%
-images of codes in Table \ref{tab:R2}. 24 new extremal binary
self-dual codes of length 68 are obtained as Gray images of the
extensions. Similar to Section \ref{64section} lifts can be applied
to the extensions. If $X$ is a
possible extension vector for a free self-dual code $\mathcal{C}$ over $%
R_{2,1}$ then $\pi _{u}\left( X\right) $ is an extension vector for
$\pi _{u}\left( \mathcal{C}\right) $. In order to extend
$\mathcal{C}$ we may
lift an extension vector for $\pi _{u}\left( \mathcal{C}\right) $. Theorem %
\ref{bound} gives an idea on which extension vectors to lift. For
instance,
a possible extension vector for the binary code $\pi _{u}\left( \mathcal{%
\varphi }_{u}\left( M_{12}\right) \right) $ is $\left(
00010111001100110000001000110011\right) $. By considering the lifts
of this vector we were able to obtain new codes with weight
enumerators corresponding to rare parameters $\gamma =4$ and $\beta
=86$, $96$ and $98$. Those are listed in Table \ref{tab:ext}.
Considering lifts reduces the workload remarkably from $4^{32}$ to
$2^{32}$.

\begin{table}[H]
\caption{New codes in $W_{68,2}$ by Theorem \protect\ref{ext} on
$R_{2,1}$ (24 codes)} \label{tab:ext}\centering
\begin{tabular}{||l|l|l|l|l||}
\hline Code & $X$ & $c$ & $\gamma $ & $\beta $ \\ \hline\hline
$\mathcal{L}_{3}$ & $\left( 31u1u11133u10u113u10u33013010111\right)
$ & $1+u $ & $0$ & $178$ \\ \hline $\mathcal{L}_{3}$ & $\left(
10u1u033uu3u00u03101010uu10u3u0u\right) $ & $1$ & $1$ & $180$ \\
\hline $\mathcal{L}_{12}$ & $\left(
11330u11u1103101u3u3101u31uu33u\right) $ & $1$ & $2$ & $164$ \\
\hline $\mathcal{L}_{8}$ & $\left(
0uuu0011113u13u01303113033311003\right) $ & $1$ & $2$ & $168$ \\
\hline $\mathcal{L}_{8}$ & $\left(
00000031313033u031u3333u33311003\right) $ & $1+u$ & $2$ & $178$ \\
\hline $\mathcal{L}_{8}$ & $\left(
u0uuu033111033uu1301113u13331uu1\right) $ & $1$ & $2$ & $180$ \\
\hline $\mathcal{L}_{8}$ & $\left(
u0u00011313u31u011u1113u33313u01\right) $ & $1$ & $2$ & $182$ \\
\hline $\mathcal{L}_{3}$ & $\left(
u3uuu33uu10uu00103010u001u030u13\right) $ & $1+u$ & $2$ & $186$ \\
\hline $\mathcal{\varphi }_{u}\left( M_{12}\right) $ & $\left(
13331031u0u1133u1111111111111111\right) $ & $1$ & $3$ & $94$ \\
\hline $\mathcal{\varphi }_{u}\left( M_{4}\right) $ & $\left(
11u301u33u0133u3u1u3u0uu010330uu\right) $ & $1$ & $3$ & $107$ \\
\hline $\mathcal{\varphi }_{u}\left( M_{12}\right) $ & $\left(
11333u3100u1133u1331133313313133\right) $ & $1+u$ & $3$ & $118$ \\
\hline $\mathcal{\varphi }_{u}\left( M_{17}\right) $ & $\left(
1310u30u330100001111111111111111\right) $ & $1+u$ & 3 & 120 \\
\hline $\mathcal{L}_{3}$ & $\left( u u u 3 1 0 u 1 1 u 3 u 0 0 u 1 u
u u 3 0 3 u 3 u 3 u 1 3 3 3 3\right) $ & $1 $ & $3$ & $164$ \\
\hline $\mathcal{L}_{3}$ & $\left( u u 0 3 1 u 0 3 1 0 3 u 0 u 0 1 u
0 0 1 0 3 u 3 u 1 u 1 3 1 1 1\right) $ & $1+u $ & $3$ & $166$ \\
\hline $\mathcal{L}_{3}$ & $\left(
uuu11u031u3u0u01u003u3u303u31131\right) $ & $1+u $ & $3$ & $168$ \\
\hline $\mathcal{L}_{3}$ & $\left( u 0 0 3 1 0 0 3 3 0 1 u u u u 3 u
0 0 1 u 1 0 3 u 3 u 3 1 3 3 1\right) $ & $1+u $ & $3$ & $174$ \\
\hline $\mathcal{\varphi }_{u}\left( M_{12}\right) $ & $\left( 0 0 0
1 0 1 3 3 0 0
1 1 u u 3 3 0 u 0 0 0 u 1 u 0 0 1 1 u u 3 3\right) $ & $1+u$ & $4$ & $86$ \\
\hline $\mathcal{\varphi }_{u}\left( M_{12}\right) $ & $\left( u u 0
1 u 1 1 1 0 u
3 3 u 0 3 3 u u 0 u 0 0 3 u 0 u 3 1 u 0 1 3\right) $ & $1+u$ & $4$ & $96$ \\
\hline $\mathcal{\varphi }_{u}\left( M_{12}\right) $ & $\left(
0001u3130u31uu1100uuuu1uu0130u31\right) $ & $1+u$ & $4$ & $98$ \\
\hline $\mathcal{\varphi }_{u}\left( M_{12}\right) $ & $\left(
u3u3u1110u3310u31111111111111111\right) $ & $1$ & $4$ & $104$ \\
\hline $\mathcal{\varphi }_{u}\left( M_{12}\right) $ & $\left(
u3010333003110031331133333113313\right) $ & 1 & $4$ & $108$ \\
\hline $\mathcal{\varphi }_{u}\left( M_{12}\right) $ & $\left(
u1u10313001110u33313331111331111\right) $ & 1 & $4$ & $112$ \\
\hline $\mathcal{L}_{8}$ & $\left(
00u00u33111u130011u31130111310u3\right) $ & $1+u$ & $4$ & $174$ \\
\hline $\mathcal{L}_{8}$ & $\left(
u300033003u0uuu10303000u1uu10u31\right) $ & $1$ & $4$ & $194$ \\
\hline
\end{tabular}%
\end{table}

\begin{remark}
The binary generator matrices of the new extremal binary self-dual
codes of
lengths $66$ and $68$ that are constructed in tables \ref{tab:66}, \ref%
{tab:idext} and \ref{tab:ext} are available online at
\cite{webpage}.
\end{remark}

\textbf{Acknowledgements}\newline \vspace{2mm}The authors would like
to thank Bahattin Y\i ld\i z for his valuable comments.


\begin{thebibliography}{99}
\bibitem{betsumiya} K. Betsumiya, S. Georgiou, T.A. Gulliver, M. Harada and
C. Koukouvinos, \textquotedblleft On self-dual codes over some prime
fields", \emph{Discrete Math}, vol 262, pp. 37--58, 2003.

\bibitem{magma} W. Bosma, J. Cannon and C. Playoust, \textquotedblleft The
Magma algebra system. I. The user language", \emph{J. Symbolic
Comput.}, vol. 24, pp. 235---265, 1997.

\bibitem{brualdi} R. A. Brualdi and V. S. Pless, \textquotedblleft Weight enumerators of
self-dual codes", \emph{IEEE Trans. Inform. Theory}, Vol. 37, pp.
1222--1225, 1991.

\bibitem{conway} J. H. Conway, N. J. A. Sloane, \textquotedblleft A new
upper bound on the minimal distance of self-dual codes", \emph{\
IEEE Trans. Inform. Theory}, Vol. 36, 6, 1319--1333, 1990.

\bibitem{dougherty} S. T. Dougherty, T. A. Gulliver, M., Harada,
\textquotedblleft Extremal binary self dual codes", \emph{IEEE
Trans. Inform. Theory}, Vol. 43 pp. 2036-2047, 1997.

\bibitem{doughertyfrobenius} S.T. Dougherty, J.L. Kim, H. Kulosman and H.
Liu, \textquotedblleft Self-dual codes over commutative Frobenius
rings", \emph{Finite Fields Appl.}, Vol.16, pp.14--26, 2010.

\bibitem{doughertyrk} S.T. Dougherty, B. Yildiz and S. Karadeniz,
\textquotedblleft Codes over $R_{k},$Gray Maps and their Binary
Images", \emph{Finite Fields Appl.}, Vol. 17, pp. 205--219, 2011.

\bibitem{gaborit} P. Gaborit, \textquotedblleft Quadratic double circulant
codes over fields", \emph{Journal of Combinatorial Theory Series A},\emph{\ }%
Vol. 97, Issue 1, pp. 85-107, 2002.

\bibitem{karadeniz} S. Karadeniz, B. Y\i ld\i z and N. Ayd\i n,
\textquotedblleft Extremal binary self-dual codes of lengths $64$
and $66$
from four-circulant constructions over codes $\mathbb{F}_{2}+u\mathbb{F}_{2}$%
",\emph{\ FILOMAT}, Vol. 28, Issue 5, pp. 937-945, 2014.

\bibitem{karadeniz66} S. Karadeniz and B.Yildiz, \textquotedblleft New
extremal binary self-dual codes of length 66 as extensions of
self-dual code over $R_{k}$", \emph{J. Franklin Inst.}, Vol. 350,
no. 8, pp.1963--1973, 2013.

\bibitem{karadeniz68} S. Karadeniz, B. Yildiz, \textquotedblleft New
extremal binary self-dual codes of length $68$ from $R_{2}$-lifts of
binary self-dual codes ", \emph{Advances in Mathematics of
Communications}, Vol. 7, no. 2, pp. 219--229, 2013.

\bibitem{kim} J.-L. Kim, \textquotedblleft New extremal self-dual codes of lengths $36,38$ and
$58$", \emph{IEEE Trans. Inf. Theory}, Vol.47, No.1, pp.386--393,
2001.

\bibitem{kaya} A. Kaya, B. Yildiz, \.{I}. \c{S}iap, \textquotedblleft New extremal binary
self-dual codes from\emph{\
}$\mathbb{F}_{4}+u\mathbb{F}_{4}$\emph{-}lifts of quadratic double
circulant codes over\emph{\ }$\mathbb{F}_{4}$", available online at
http://arxiv.org/abs/1405.7147

\bibitem{jacodes} A. Kaya, B. Y\i ld\i z, \textquotedblleft New extremal binary self-dual
codes of length $68$", \emph{Journal of Algebra Combinatorics
Discrete Structures and Applications}, Vol. 1, No. 1, pp 29-39,
2014.

\bibitem{webpage} A. Kaya, N. T\"{u}fek\c{c}i, \emph{Binary generator
matrices of new extremal self-dual binary codes of lengths 66 and
68},
available online at \href{url}{http://www.fatih.edu.tr/\symbol{126}%
akaya/newbinary66-68.html}

\bibitem{kayayildiz} A. Kaya, B. Yildiz, \textquotedblleft Extension theorems for
self-dual codes over rings and new binary self-dual codes",
available online at http://arxiv.org/abs/1404.0195.

\bibitem{rains} E. M. Rains, \textquotedblleft Shadow Bounds for Self Dual
Codes", \emph{IEEE Trans. Inf. Theory}, Vol.44, pp.134--139, 1998.

\bibitem{shi} M. Shi, L. Chen \textquotedblleft Construction of two-Lee
weight codes over
$\mathbb{F}_{p}+v\mathbb{F}_{p}+v^{2}\mathbb{F}_{p}$ ",
\emph{International Journal of Computer Mathematics}.

\bibitem{tufekci} N. T\"{u}fek\c{c}i, B. Y\i ld\i z, \textquotedblleft On codes over $R_{k,m}$
and constructions for new binary self-dual codes", to appear in \emph{%
Mathematica Slovaca.}

\bibitem{tsai} H-P. Tsai, P-Y. Shih, R-Y. Wuh, W-K. Su, C-H. Chen,
\textquotedblleft Construction of self-dual codes", \emph{\ IEEE
Trans. Inform. Theory}, Vol. 54, pp. 3826--3831, 2008.

\bibitem{wood} J. Wood, \textquotedblleft Duality for modules over finite rings and
applications to coding theory", \emph{Amer. J. Math.}, Vol. 121, pp.
555--575, 1999.

\bibitem{yankov} N. Yankov, \textquotedblleft Self-dual $\left[ 62,31,12%
\right] $ and $\left[ 64,32,12\right] $ codes with an automorphism
of order 7", \emph{Advances in Mathematics of Communications},
Vol.8, No1. pp.73--81, 2014.

\bibitem{yildizr2} B. Yildiz, S. Karadeniz, \textquotedblleft Linear Codes
over
$\mathbb{F}_{2}+u\mathbb{F}_{2}+v\mathbb{F}_{2}+uv\mathbb{F}_{2}$",
\emph{Des. Codes Crypt}. vol.54, pp. 61--81, 2010.
\end{thebibliography}
\end{document}